\documentclass{article}
\usepackage{amsmath,amsfonts,amssymb,amsthm}
\usepackage[round]{natbib}
\usepackage{booktabs}
\usepackage{graphicx}
\usepackage{url}

\begin{document}

\title{Moment based gene set tests}
\date{April 2014}



\author{Jessica L. Larson\\ Genentech Inc.
\and
Art B. Owen\\ Stanford University}


\renewcommand{\le}{\leqslant}
\renewcommand{\ge}{\geqslant}

\newcommand{\wt}{\widetilde}
\newcommand{\wh}{\widehat}

\newcommand{\e}{\mathbb{E}}

\newcommand{\var}{\mathrm{var}}
\newcommand{\cov}{\mathrm{cov}}

\newcommand{\dbeta}{\mathrm{beta}}

\newcommand{\tran}{\mathsf{T}}

\newcommand{\yi}{Y_i}
\newcommand{\wtyi}{\wt Y_i}
\newcommand{\yip}{Y_{i'}}
\newcommand{\wtyip}{\wt Y_{i'}}

\newcommand{\wtyj}{\wt Y_j}
\newcommand{\wtyk}{\wt Y_k}
\newcommand{\wtyl}{\wt Y_\ell}

\newcommand{\pe}{\phantom{=}}
\newcommand{\pn}{\phantom{-}}

\newcommand{\sums}{\sideset{}{^*}\sum}

\renewcommand{\emptyset}{\varnothing}

\newcommand{\dnorm}{\mathcal{N}}
\newcommand{\real}{\mathbb{R}}

\newcommand{\simiid}{\stackrel{\mathrm{iid}}\sim}
\newcommand{\dotsim}{\stackrel{\cdot}\sim}
\newcommand{\phz}{\phantom{0}}

\newtheorem{lemma}{Lemma}
\newtheorem{corollary}{Corollary}

\theoremstyle{definition}
\newtheorem{example}{Example}


\maketitle

\begin{abstract}
{\bf Motivation:}
Permutation-based gene set tests are standard approaches for testing relationships 
between collections of related genes and an outcome of interest 
in high throughput expression analyses. Using $M$ random permutations, 
one can attain $p$-values as small as $1/(M+1)$.
When many gene sets are tested, we need smaller
$p$-values, hence larger $M$, to achieve significance while accounting for the number of simultaneous
tests being made.  As a result, the number of permutations to be done
rises along with the cost per permutation.  To reduce this cost, we
seek parametric approximations to the permutation distributions for gene set tests.

{\bf Results:}
We focus on two gene set methods related to sums and sums of squared
$t$ statistics. Our approach calculates exact relevant moments of a 
weighted sum of (squared)
test statistics under permutation.  We find moment-based gene set enrichment 
$p$-values that closely approximate the permutation method $p$-values.
The computational cost of our algorithm for linear statistics
is on the order of doing $|G|$ permutations, where $|G|$ 
is the number of genes in set $G$.
For the quadratic statistics, the cost is on the order of $|G|^2$
permutations which is orders of magnitude faster
than naive permutation.
We applied the permutation approximation method to three 
public Parkinson's Disease expression datasets 
and discovered enriched gene sets not previously discussed. 
In the analysis of these experiments with our method, we are able to 
remove the granularity effects of permutation
analyses and have a substantial computational speedup with little cost to accuracy.

{\bf Availability:}
Methods available as a Bioconductor package, npGSEA (www.bioconductor.org).

{\bf Contact:} {larson.jessica@gene.com} 
\end{abstract}

\section{Introduction}
In a genome-wide expression study, researchers often compare the level of gene expression in 
thousands of genes between two treatments groups (e.g., disease, drug, genotype, etc.).  
Many individual genes may trend toward differential expression, but will 
often fail to achieve significance.  This could
happen for a set of genes in a given pathway or system (a gene set).  A number
of significant and related genes taken together can provide
strong evidence of an association between the corresponding gene set and treatment of interest.  
Gene set methods can improve power by looking for small, coordinated expression
changes in a collection of related genes, rather than testing for large shifts in many individual genes.  

Additionally, single gene methods often assume that all genes are independent of each other; this is 
not likely true in real biological systems.  
With known gene sets of interest, researchers can use existing biological knowledge to drive their 
analysis of genome-wide expression data,
thereby increasing the interpretability of their results.

\cite{moothetal:2003} first introduced gene set enrichment analysis (GSEA) 
and calculated gene set $p$-values based on Kolmogorov-Smirnov statistics.  
Since then, there have been many methodological proposals for GSEA;
no single one is always the best.  For example, some tests
are better for a large number of weakly associated genes,
while others have better power for a small number of strongly
associated genes \citep{newt:quin:fern:denb:seng:ahlq:2007}.

One of the most important differences among gene set methods
is the definition of the null hypothesis.  
Tian {\em et~al.}, 2005 and Goeman and B\"uhlmann, 2007 (among others)
introduce two null hypotheses that differentiate the general approaches for gene set methods. 
The first measures whether a gene set is more strongly related with the outcome of interest 
than a comparably sized gene set.  
Methods of this type typically rely on randomizing the gene
labels to test what is often called the {\em competitive}
null hypothesis.  This is problematic
because genes are inherently correlated (especially those within a set) and permuting them
does not give a rigorous test \citep{goem:buhl:2007}. 

The second type of approach is used to determine whether the genes within a 
set associate more strongly with the outcome of interest than they would by chance, 
had they been independent of the 
outcome.  Methods that test this 
{\em self-contained} null hypothesis usually judge statistical significance by 
randomizing the phenotype with respect to expression data and assume that 
gene sets are fixed.  
While we acknowledge that the {\em competitive} hypothesis is often of interest, 
we focus on methods that test the {\em self-contained} hypothesis in this paper.

A popular {\em self-contained} GSEA method is the JG-score 
 \citep{jian:gent:2007}, which determines the the level of enrichment based on 
averaging linear model statistics.
Recently, \cite{acke:stri:2009} compared $261$ different gene set tests, and 
found particularly good performance from a sum of squared
single gene regression coefficients.  
We extend both the sum and the sum of 
squared linear statistics approaches 
with a new method in this paper. 

All current GSEA methods are based on permutation approaches.  The initial 
GSEA \citep{moothetal:2003} and JG-score \citep{jian:gent:2007} methods both have 
closed form null distributions for their enrichment statistics, 
Gaussian and Kolmogorov-Smirnov, respectively; 
however, even the authors of these methods acknowledge that these distributions do 
not give the correct $p$-values and 
suggest the use of permutation.  
\cite{lehm:roma:2005} give a 
concise explanation of how permutation inference works.  
It is common to approximate the permutation distribution by
a large Monte Carlo sample (Eden and Yates, 1933; David, 2008).

Permutation tests
are simple to program and do not make parametric distributional assumptions.
They also can be applied to almost any statistic we might wish to investigate.  
However, permutation approaches are often computationally expensive, 
are subject to random inference, and fail to 
achieve continuous $p$-values.
Each of these drawbacks is described 
in more depth below.  

We have developed a new gene set enrichment approach 
that approximates the permutation distribution of our corresponding 
test statistics.  We find that our method of moments techniques 
result in almost exactly the same $p$-values as permutation approaches,
but in much less computation time.  Through our approach, we are able to 
obtain refined $p$-values and achieve stringent significance thresholds.  
We applied our approach to three 
public expression analyses, and found disease-associated gene sets not 
previously discovered in these studies.

\section{Methods}\label{sec:methods}
\subsection{The data}
For definiteness, we present our notation using the language of gene expression
experiments.  
Let $g$,  $h$, $r$, and $s$ denote individual genes
and $G$ be a set of genes.
The cardinality of $G$ is denoted $|G|$, or sometimes $p$.
That is the same letter we use for $p$-value, but the usages
are distinct enough that there should be no confusion.
Our experiment has $n$ subjects.
The subjects may represent patients, cell cultures, or tissue samples.

The expression level for gene $g$ in subject $i$ 
is $X_{gi}$, and $Y_{i}$ is the target variable on subject $i$. 
$Y_{i}$ is often a treatment, disease, or other phenotype.
We center the variables so that
\begin{align}\label{eq:centered}
\sum_{i=1}^nY_{i} = \sum_{i=1}^n X_{gi} = 0,\quad\forall g.
\end{align}

The $X_{gi}$ are not necessarily raw expression values, nor are they restricted to
microarray values.
In addition to the centering~\eqref{eq:centered} they could
have been scaled to have a given mean square. The scaling
factor for $X_{gi}$ might even depend on the sample variance
for some genes $h\ne g$ if we thought that shrinking
the variance for gene $j$ towards the others
would yield a more stable test statistic \citep{smyth:2005}.  We might equally use a quantile
transformation, replacing the $j'$th largest
of the raw $X_{gi}$ by $\Phi^{-1}( (j-1/2)/n)$ where
$\Phi$ is the Gaussian cumulative distribution function.
Further preprocessing may be advised to handle
outliers in $X$ or $Y$. We do require that the preprocessing
of the $X$'s does not depend on the $Y$'s and vice versa.


\subsection{Test statistics}
Our measure of association for gene $g$ on our treatment of interest is 
\begin{align}\label{eq:defbeta}
\hat\beta_{g} = \frac1{n}\sum_{i=1}^nX_{gi}Y_{i}.
\end{align}
If both $X_{gi}$ and $Y_i$ are centered and standardized
to have variance $1$, then $\hat\beta_g=\hat\rho_g$,
the sample correlation between $Y$ and gene $g$.
The usual $t$-statistic for testing a linear relationship
between these variables is $t_g \equiv \sqrt{n-2}\hat\rho_g/(1-\hat\rho_g^2)^{1/2}$,
which is a monotone transformation of $\hat\rho_g$.

For reasons of power and interpretability, we apply gene
set testing methods instead of just testing individual
genes.  Linear and quadratic test statistics have been 
found to be the best performers for
gene set enrichment analyses; we thus consider two statistics for our approach:
\begin{align*}
\wh T_{G,w} & = \sum_{g\in G}w_{g}\hat\beta_{g}\quad\text{and}\quad
\wh C_{G,w} = \sum_{g\in G}w_{g}\hat\beta_{g}^2.
\end{align*}

The statistic $\wh T_{G,w}$ can
approximate the JG score of \cite{jian:gent:2007}.
The JG score is
$(1/\sqrt{|G|})\sum_{g\in G}t_g$.
Taking $w_g = \sqrt{n-2}/(\mathrm{sd}(X_g)\mathrm{sd}(Y))$, 
where $\mathrm{sd}$ denotes a standard deviation,
weights genes similarly to the JG score.  Although $\wh T_{G,w}$ with
these weights sums statistics
equivalent to $t$ statistics, it is not exactly equivalent to the
sum of those statistics because of the way $\hat\rho_g$
appears in the denominator of each $t_g$.

The statistic $\wh C_{G,w}$ is a weighted sum of squared
sample covariances. \cite{acke:stri:2009}
conducted an extensive simulation of gene set methods
and found good results for quadratic combinations of per
gene test statistics.

The letters $T$ and $C$ are mnemonics for the $t$ and $\chi^2$ distributions
that resemble the permutation distributions of these quantities.
The $w_{g}$ are scalar weights.  For the quadratic statistics
we will suppose that $w_{g}\ge 0$.  We won't need that condition to find
moments of $C_{G,w}$, but because we will compare $C_{G,w}$ to a $\chi^2$ distribution,
it is reasonable to avoid negative weights.  Non-negative weights
are also used to simplify our algorithm.

Although linear and quadratic test statistics are fairly restricted,
they do allow a reasonable amount of customization through the weights
$w_g$, and they are very interpretable compared to more ad hoc statistics.

\subsection{Permutation procedure}

A permutation of $\{1,2,\dots,n\}$ is 
a reordering of $\{1,2,\dots,n\}$.
There are $n!$ permutations. We call $\pi$ a
\emph{uniform random permutation} of $\{1,2,\dots,n\}$
if it equals each distinct permutation with probability $1/n!$.

In a permutation analysis, we replace $Y_{i}$ by $\wt Y_{i}$
where $\wt Y_{i} = Y_{\pi(i)}$ for $i=1,\dots,n$.
Then $\wt\beta_{g} = (1/n)\sum_{i=1}^nX_{gi}\wt Y_{i}$, 
and when $\wt Y$ is substituted for $Y$, 
$\wh T_{G,w}$ becomes $\wt T_{G,w}$ and
$\wh C_{G,w}$ becomes $\wt C_{G,w}$.


The $n!$ different permutations form a reference distribution
from which we can compute $p$-values.
There are often so many possible permutations that we cannot calculate or use all of them.
Instead, we independently sample uniform random permutations $M$
times, 
getting statistics $\wt C_m=\wt C_{G,w,m}$, and similarly $\wt T_m$,
for $m=1,\dots,M$.
We then compute $p$-values by comparing our observed statistics to our permutation distribution:
\begin{align*}
p_Q & = \frac{\#\{\wt C_m\ge \wh C\}+1}{M+1}
& p_C & = \frac{\#\{|\wt T_m|\ge |\wh T|\}+1}{M+1}\\
p_L & = \frac{\#\{\wt T_m\le \wh T\}+1}{M+1}, 
\quad\text{or} &p_R & = \frac{\#\{\wt T_m\ge \wh T\}+1}{M+1},
\end{align*}
where $p_Q$ and $p_C$ are $p$-values for two-sided inferences on the 
quadratic and linear statistic, respectively, and 
$p_L$ (left) and $p_R$ (right) are for one-sided inferences 
based on the linear statistic.  We use the mnemonic $C$ in $p_C$ 
to denote the central (or two-sided) $p$-value, which corresponds 
to a central confidence interval.
The $+1$ in numerator and denominator of the $p$-values
corresponds to counting the sample test statistic as one
of the permutations.  That is, we automatically include an
identity permutation. 


\subsection{Permutation disadvantages}
There are three main disadvantages to permutation-based analyses: cost,
randomness, and granularity.

Testing many sets of genes 
becomes computationally expensive for two reasons.  
First, there are many test statistics to calculate
in each permuted version of the data. 
Second, to allow for multiplicity adjustment, we
require small nominal $p$-values to draw inference about our sets, 
which in turn requires a large number
of permutations.  That is, to obtain a small adjusted $p$-value (e.g., via 
FDR, FWER, Bonferroni methods), 
one first needs a small enough raw $p$-value.  In order 
to obtain small raw $p$-values, the number of permutations ($M$) must be large, 
thereby increasing computational cost.

Because permutations are based on a random shuffling of the data,
there is a chance that we will obtain a different $p$-value 
for our set of interest each 
time we run our permutation analysis.  That is, our inference is subject to a given 
random seed.

Permutations also have a granularity problem.
If we do $M$ permutations, then the smallest possible
$p$-value we can attain is $1/(M+1)$.  At or below this minimum $p$-value 
permutation tests have no power.
\cite{knij:wess:rein:shmu:2009}
suggest that for a reliable $p$-value, there should be at least $10$ permuted
values more extreme than the sample.
That requires $M\approx10/p$ and when it is necessary,
due to test multiplicity, to use small $p$ such as $10^{-6}$ 
or smaller, the permutation approach becomes computationally expensive.  
We call this the \emph{sample granularity} problem.

There is also a \emph{population granularity} problem. 
In an experiment with $n$ observations, the smallest
possible $p$-value is at least $1/n!$. Sometimes the attainable
minimum is much larger.
For instance, when the target variable $Y$ is
binary with $n/2$ positive and $n/2$ negative values
then the smallest possible $p$-value is 
$1/{n\choose n/2}$. For $n=10$ we necessarily have $p\ge 1/252$.  
Rotation sampling methods such as ROAST are able to get 
around this population granularity problem \citep{wu:lim:vail:asse:visv:smyt:2010}.
Increased Monte Carlo sampling can mitigate the sample granularity
problem but not the population granularity problem.

Another aspect of the granularity problem is that permutations
give us no basis to distinguish between two gene sets that
both have the same $p$-value $1/(M+1)$. There may be many such gene sets,
and they have meaningfully different effect sizes.  Many current approaches 
solve this problem by ranking significantly enriched gene sets by  
their corresponding test statistics.  This practice only works if all test statistics have 
the same null distribution and correlation structure, which is not the case for many 
current GSEA methods.  Additionally, the resulting broken ties do not have a $p$-value 
interpretation and cannot be directly used in multiple testing methods.  To break ties in this way 
also requires the retention of both a $p$-value and a test statistic for inference, rather 
than just one value.

Because of each of these limitations of permutation testing, there is a need to move beyond
permutation-based GSEA methods.  The methods we present below are not as computationally 
expensive, random, or granular as their permutation counterparts.  Our proposal results in a 
single number on the $p$-value scale.

\subsection{Moment based reference distributions}
To avoid the issues discussed above, we approximate the distribution of
the permuted test statistics $\wt T_{G,w}$ by Gaussians
or by rescaled beta distributions.
For quadratic statistics $\wt C_{G,w}$ we use
a distribution of the form $\sigma^2\chi^2_{(\nu)}$ choosing
$\sigma^2$ and $\nu$ to match the second and fourth moments of $\wt C_{G,w}$
under permutation.  

For the Gaussian treatment of $\wt T_{G,w}$ we
find $\sigma^2 = \var(\wt T_{G,w})$ under permutation
using equation~\eqref{eq:varT} of Section~\ref{sec:costs}
and then report the $p$-value
$$
p = \Pr( \dnorm( 0, \sigma^2 ) \le \wh T_{G,w}),
$$
where $\wh T_{G,w}$ is the observed value of the linear statistic.
The above is a left tail $p$-value. Two-tailed and right-tailed $p$
values are analogous.

When we want something sharper than the normal distribution, we can
use a scaled Beta distribution, of the form 
$A + (B-A)\dbeta(\alpha,\beta)$. 
The $\dbeta(\alpha,\beta)$ distribution has 
a continuous density function
on $0<x<1$ for $\alpha,\beta>0$.
We choose $A$, $B$, $\alpha$ and $\beta$ by matching the upper
and lower limits of $\wt T_{G,w}$, as well as its mean and variance.
Using equation~\eqref{eq:varT}
from our theory section we have
\begin{equation}
\begin{split}
A & = \min_{\pi} \frac1n\sum_{i=1}^n \sum_{g\in G}w_gX_{gi} Y_{\pi(i)},\\
B & = \max_{\pi} \frac1n\sum_{i=1}^n \sum_{g\in G}w_gX_{gi} Y_{\pi(i)},\\
\alpha & = \frac{A}{B-A}\Bigl( \frac{AB}{\var(\wt T_{G,w})}+1\Bigr),\quad\text{and}\\
\beta & = \frac{-B}{B-A}\Bigl( \frac{AB}{\var(\wt T_{G,w})}+1\Bigr).
\end{split}
\end{equation}
The observed left-tailed $p$-value is
$$
p = \Pr\Bigl( \dbeta(\alpha,\beta) \le \frac{\wh T_{G,w}-A}{B-A}\Bigr).
$$

It is easy to find the permutations that maximize and minimize
$\wt T_{G,w}$ by sorting the $X$ and $Y$ values appropriately
as described in Section~\ref{sec:costs}.
The result has $A < 0<B$.  For the beta distribution to have
valid parameters we must have $\sigma^2<-AB$.
From the inequality of \cite{bhat:chan:2000}, we
know that $\sigma^2\le -AB$.  There are in fact degenerate
cases with $\sigma^2=-AB$, but in these cases $\wt T_{G,w}$ only
takes one or two distinct values under permutation, and those
cases are not of practical interest.

Like us, \nocite{zhou:wang:wang:2009}
Zhou et al. (2009) have used a beta distribution to approximate a permutation.  
They used the first 4 moments of a Pearson curve for their approach.  
Fitting by moments in the Pearson family, 
it is possible to get a beta distribution whose support set $(A,B)$ does
not even include the observed value.  That is, the
observed value is even more extreme than it would
have to be to get $p=0$; it is almost like getting $p<0$.  
We chose $(A,B)$ based on the upper
and lower limits of $\wt T_{G,w}$ to prevent our observed 
test statistic from falling outside the range of possible values of 
our reference distribution (Section~\ref{sec:costs}).


For the quadratic test statistic $\wh C_{G,w}$
we use a $\sigma^2\chi^2_{(\nu)}$ reference distribution
reporting the two-tailed $p$-value
$\Pr( \sigma^2\chi^2_{(\nu)}\ge \wh C_{G,w})$
after matching the first and second moments of $\sigma^2\chi^2_{(\nu)}$
to $\e(\wt C_{G,w})$ and $\e( \wt C_{G,w}^2)$ respectively.
The parameter values are
$$
\nu = 2\frac{\e(\wt C_{G,w})^2}{\var(\wt C_{G,w})}
\quad\text{and}\quad \sigma^2=\frac{\e(\wt C_{G,w})}{\nu}
=\frac{\var(\wt C_{G,w})}{2\e(\wt C_{G,w})}.
$$
Our formulas for $\e(\wt C_{G,w})$ and $\e(\wt C_{G,w}^2)$ under permutation
are given in equation~\eqref{eq:targetmoments}
of Section~\ref{sec:lemmas}.
Those formulas use $\e(\wt\beta_g^2)$ 
and $\cov(\wt\beta_g^2,\wt\beta_h^2)$
which we give in Corollaries~\ref{cor:cov} and~\ref{cor:cov2} 
of Section~\ref{sec:lemmas}.

All of our reference
distributions are continuous and unbounded and hence
they avoid the granularity problem of permutation testing.
We have prepared a publicly available Bioconductor \citep{gentleman:2004} package, npGSEA, 
which implements our algorithm and calculates the corresponding statistics discussed 
in this section.

\section{Theoretical results}\label{sec:theory}

\subsection{Permutation moments of test statistics}\label{sec:lemmas}

Under permutation, $\e(\wt Y_{i})=0$ by symmetry,
and so  $\e(\wt \beta_{g})=0$ too.
We easily find that,
\begin{equation}\label{eq:targetmoments}
\begin{split}
\e( \wt T_{G,w} ) & = 0,\\
\var( \wt T_{G,w} ) & = \sum_{g\in G}\sum_{h\in G} w_{g}w_{h}\cov(\wt\beta_{g},\wt\beta_{h})\\
\e( \wt C_{G,w})
&= \sum_{g\in G}w_{g}\e(\wt\beta_{g}^2),\quad\text{and}\\
\var( \wt C_{G,w})
&= \sum_{g\in G}\sum_{h\in G}
w_{g}w_{h}\cov( \wt\beta^2_{g},\wt\beta^2_{h}).
\end{split}
\end{equation}
The means, variances and covariances in~\eqref{eq:targetmoments}
are taken with respect to the random permutations with the data $X$ and $Y$
held fixed. We adopt the convention that moments of permuted
quantities are taken with respect to the permutation 
and are conditional on the $X$'s and $Y$'s.
This avoids cumbersome
expressions like $\e(\wt\beta^2_{g}\mid X_{gi}, Y_{i}, g\in G)$. 

We will need the following even moments of $X$ and $Y$:
\begin{align*}
\mu_2 & = \frac1n\sum_{i=1}^n Y_{i}^2,\quad
\mu_4  = \frac1n\sum_{i=1}^n Y_{i}^4,\\
\bar X_{gh} &= \frac1n\sum_{i=1}^n X_{gi}X_{hi},\quad\text{and}\\
\bar X_{ghrs} & = \frac1n\sum_{i=1}^n X_{gi}X_{hi}X_{ri}X_{si}
\end{align*}
for $g,h,r,s\in G$. 
Although our derivations involve $O(p^4)$ different moments when the
gene set $G$ has $p$ genes, our computations do not require all of those
moments.

\begin{lemma}\label{lem:moment1}
For an experiment with $n\ge2$ including genes $g$ and $h$,
$$\e(\wt\beta_g\wt\beta_h) = \frac{\mu_2\bar X_{gh}}{n-1}.$$
\end{lemma}
\begin{proof}
See Appendix 1.\quad$\Box$
\end{proof}

\begin{corollary}\label{cor:cov}
For an experiment with $n\ge 2$ including genes $g$ and $h$,
$$\cov(\wt\beta_g,\wt\beta_h)=\mu_2\bar X_{gh}/(n-1).$$
\end{corollary}
\begin{proof}
This follows from Lemma~\ref{lem:moment1} because $\e(\wt\beta_g)=0$.
\end{proof}

From Corollary~\ref{cor:cov}, we see that the correlation
between permuted test statistics $\wt\beta_g$ and $\wt\beta_h$
is simply the correlation between expression values for genes $g$
and $h$.

\begin{lemma}\label{lem:moment2}
For an experiment with $n\ge4$ including genes $g,h,r,s$,
$$
\e( \wt\beta_g\wt\beta_h\wt\beta_r\wt\beta_s )
=
\begin{pmatrix} \mu_2^2\\\mu_4\end{pmatrix}^\tran 
A^\tran B
\begin{pmatrix}
\bar X^*_{ghrs}
/n^2\\[1ex]
\bar X_{ghrs}/n^3
\end{pmatrix} 
$$
where $\bar X^*_{ghrs}
=\bar X_{gh}\bar X_{rs}+\bar X_{gs}\bar X_{hr}+\bar X_{gr}\bar X_{hs}$,
with $A^\tran$ given by
\begin{align*}
\begin{pmatrix}
0 & 0 & \dfrac{n}{n-1} & \dfrac{-n}{(n-1)(n-2)} & \dfrac{3n}{(n-1)(n-2)(n-3)}\\[2ex]
1 & \dfrac{-1}{n-1} & \dfrac{-1}{n-1} & \dfrac2{(n-1)(n-2)} & \dfrac{-6}{(n-1)(n-2)(n-3)}
\end{pmatrix},
\end{align*}
and
\begin{align*}
B  &= 
\begin{pmatrix}
\pn0 &  \pn1\\
\pn0 & -4\\
\pn1 & -3\\
-2   & \,12\\
\pn1 & -6
\end{pmatrix}.
\end{align*}
\end{lemma}
\begin{proof}
See Appendix 2.\quad$\Box$
\end{proof}

The expression is complicated, but it is simple
to compute; we need only two moments of $Y$, two
cross-moments of $X$, and the $2\times 2$ matrix $A^\tran B$.
The matrix $A$ depends on the experiment
through $n$. 
Using Lemma~\ref{lem:moment2} we can obtain
the covariance between $\wt\beta^2_g$ and $\wt\beta^2_h$.

\begin{corollary}\label{cor:cov2}
For an experiment with $n\ge 4$ and genes $g,h$,
\begin{align*}
\cov(\wt\beta^2_g,\wt\beta^2_h) 
&=
\begin{pmatrix} \mu_2^2\\\mu_4\end{pmatrix}^\tran 
A^\tran B
\begin{pmatrix}
\bar X^*_{gghh}
/n^2\\[1ex]
\bar X_{gghh}/n^3
\end{pmatrix} -\frac{\mu_2^2}{(n-1)^2}\bar X_{gg}\bar X_{hh},
\end{align*}
where $\bar X^*_{gghh} = \bar X_{gg}\bar X_{hh} + 2\bar X_{gh}^2$
with $A$ and $B$ as given in Lemma~\ref{lem:moment2}.
\end{corollary}
\begin{proof}
The covariance is
$\e(\wt\beta^2_g\wt\beta^2_h)-\e(\wt\beta^2_g) \e(\wt\beta^2_h)$.
Applying Lemma~\ref{lem:moment2} to the first expectation
and Lemma~\ref{lem:moment1} to the other two yields the result.
\quad$\Box$
\end{proof}


\subsection{Rotation moments of test statistics}
Rotation sampling~\citep{wedd:1975,lang:2005} 
provides an alternative to permutations, and is justified
if either $X$ or $Y$ has a Gaussian distribution.
It is simplest to describe when $Y\sim\dnorm(\mu,\sigma^2 I_n)$
and even simpler for $Y\sim\dnorm(0,\sigma^2 I_n)$.
In the latter case we can replace $Y$ by $\wt Y = QY$
where $Q\in\real^{n\times n}$ is a random orthogonal
matrix (independent of both $X$ and $Y$), and the distribution of
our test statistics is unchanged under the null hypothesis that
$X$ and $Y$ are independent.

Rotation tests work by repeatedly sampling from the uniform distribution
on random orthogonal matrices and recomputing the test statistics using
$\wt Y$ instead of $Y$.  They suffer from sample granularity but not population
granularity because $Q$ has a continuous distribution (for $n\ge 2$).

To take account of centering we need to use a rotation test
appropriate for  $Y\sim\dnorm(\mu,\sigma^2 I_n)$.
\cite{lang:2005} does this by choosing rotation matrices that
leave the population mean fixed.  He rotates the data in an $n-1$ dimensional
space orthogonal to the vector $1_n$. To get such a rotation matrix,
he first selects an orthogonal contrast matrix
$W\in\real^{n\times (n-1)}$.  This matrix satisfies $W^\tran W=I_{n-1}$ and
$W^\tran 1_n=0_{n-1}$. Then he generates a uniform random rotation 
$Q^*\in\real^{(n-1)\times (n-1)}$ and delivers $\wt Y=QY$,
where $Q = \frac1n1_n1_n^\tran +WQ^*W^\tran$.
More generally if $Y\sim\dnorm(Z\gamma,\sigma^2 I_n)$, for a linear model
$Z\gamma$, \cite{lang:2005} shows how to rotate $Y$ in the residual space
of this model, leaving the fits unchanged.

\cite{wu:lim:vail:asse:visv:smyt:2010} have implemented rotation sampling for
microarray experiments in their method, ROAST.
They speed up the sampling by generating a random vector instead of
a random matrix.
For some tests, permutations and rotations have the same moments,
and so our approximations are approximations of rotation tests
as much as of permutation tests.

Our rotation method approximation performs very similarly to the permutation
method.  We let $\wt Y = QY$ for $Q = (\frac1n1_n1_n^\tran + WQ^*W^\tran)$
where $Q^*$ is a uniform random $n-1\times n-1$ rotation matrix
and the contrast matrix $W\in\real^{n\times(n-1)}$ 
satisfies $W^\tran 1_n=0_{n-1}$ and $W^\tran W = I_{n-1}$
and then $\wt\beta$, $\wt T$ and $\wt C$ are defined as for permutations,
substituting $\wt Y$ for $Y$.

The variance of the quadratic test statistic
depends on \emph{which} contrast matrix $W$ one chooses, 
and it cannot always match the permutation variance.
This difference disappears asymptotically as $n\to\infty$.  

\begin{lemma}\label{lem:rotate1}
For an experiment with $n\ge2$ including genes $g$ and $h$,
the moments $\e(\wt\beta_g)$ and $\e(\wt\beta_g\wt\beta_h)$
are identical to their permutation counterparts,
regardless of the choice for $W$.
\end{lemma}
\begin{proof}
See Appendix 3 and 4.\quad$\Box$
\end{proof}

\begin{corollary}
For an experiment with $n\ge 2$, 
$\e(\wt T_{G,w})$, $\var(\wt T_{G,w})$ and $\e( \wt C_{G,w})$
are the same whether $\wt Y$ is formed by permutation or
rotation of $Y$.
\end{corollary}


\subsection{Computation and costs}\label{sec:costs}

To facilitate computation for the
linear statistic, we reduce each gene set to a single pseudo-gene 
$X_{Gi} = \sum_{g\in G}w_{g}X_{gi}$ 
and then let
$$\bar X_{G} = \frac1{n}\sum_{i=1}^n X_{Gi}\quad\text{and}\quad
\bar X_{GG} = \frac1{n}\sum_{i=1}^n X_{Gi}^2.$$
The weights $w$ have been absorbed into the pseudo-gene to simplify notation.
We define
\begin{align*}
\hat\beta_{G} &= \sum_{g\in G}w_{g} \hat\beta_{g}
= \frac1n\sum_{i}X_{Gi}Y_i,\quad\text{and}\\
\wt\beta_{G} &= \sum_{g\in G}w_{g} \wt\beta_{g}
= \frac1n\sum_{i}X_{Gi}\wt Y_i.
\end{align*}

Our permuted linear test statistic
is $\wt T_{G,w} = \wt\beta_{G}$, with
\begin{align}\label{eq:varT}
\var(\wt T_{G,w})
&= \var( \wt\beta_{G})
 = \frac{\mu_{2}}{n-1}\bar X_{GG}.
\end{align}

For the beta approximation, we need the range of $\wt T_{G,w}$. 
Let the sorted $Y$ values be $Y_{(1)}\le Y_{(2)}\le \dots\le Y_{(n)}$
and the sorted $X_{Gi}$ values be $X_{G(1)}\le X_{G(2)}\le\dots\le X_{G(n)}$.
Then the range of $\wt T_{G,w}$ is $[A,B]$, where
\begin{align*}
A & = \frac1n \sum_{i=1}^n X_{G(i)}Y_{(n+1-i)},\quad\text{and}\quad
B  = \frac1n \sum_{i=1}^n X_{G(i)}Y_{(i)}.
\end{align*}

For a $\sigma t_{(\nu)}$ reference distribution
we would also need
$\e(\wt T_{G,w}^4)=\e(\wt\beta_{G}^4)$.
We can apply Lemma~\ref{lem:moment2}  to
the pseudo-gene resulting in
\begin{align}\label{eq:fourthmomentbg}
\e(\wt\beta_{G}^4) = 
\begin{pmatrix}
\mu_{2}^2\\
\mu_{4}
\end{pmatrix} A^\tran B
\begin{pmatrix}
3\bar X_{GG}^2/n^2\\
\bar X_{GGGG}/n^3
\end{pmatrix},
\end{align}
where $\bar X_{GGGG} = \frac1{n}\sum_{i=1}^n X_{Gi}^4$. 

We considered using a $\sigma t_{(\nu)}$ reference distribution for $\wt T_{G,w}$,
taking into account the fourth moment of $\wt T_{G,w}$ ~\eqref{eq:fourthmomentbg}.  
We have often (in fact usually) found
that $\e(\wt T_{G,w}^4)<3\e(\wt T_{G,w}^2)^2$; that is, lighter tails than the normal.  
This implies a negative kurtosis
for the permutation distribution, and $t$ distributions have
positive kurtosis.  For this reason we use a beta approximation
and not a $t$ approximation.

For the quadratic statistic we have found it useful
to replace $X_{gi}$ by $\sqrt{w_g}X_{gi}$ in precomputation.
That step is only valid for non-negative $w_g$, but those
are the ones of most interest.
Then we use formulas for $\e(\wt C_{G,w})$ and
$\var(\wt C_{G,w})$ with all $w_g=w_h=1$ ~\eqref{eq:targetmoments}.

Now we consider the computational cost.  
The cost to compute all of the $X_{Gi}$ is dominated by
$np$ multiplications.  It then takes $n$ more multiplications
to get $\hat\beta_{G}$ and another $n$ to get $\bar X_{GGe}$.
It costs $n$ multiplications to get all the $\mu_{2}$, except
that step done once can be used for all gene sets.
The cost for the Gaussian
approximation $\dnorm(0,\var(\wt T_{G,w}))$ is dominated by $n(p+2)$ multiplications. 

For the beta approximation there is also a cost proportional to $n\log(n)$
in the sorting to compute limits $A$ and $B$. That adds a cost comparable
to a multiple of $\log(n)$ permutations. We judge that the cost of sorting
is usually minor for $n$ and $p$ of interest in bioinformatics.

A permutation analysis requires $nM$ multiplications,
after computing $X_{Gi}$, for a total of $n(M+p)$.
It is very common for $p$ to be a few tens and $M$ to be many thousands
or more.   Then we can simplify the costs to $n(M+p)\approx nM$
and $n(2+p)\approx np$.
The moment method costs about as much
as doing $p$ permutations.  When the gene set has tens of
genes and the permutation method uses many thousands 
or even several million permutations, the computational cost is quite large.

The pseudo-gene technique is more expensive for the quadratic
statistics. The dominant cost in
computing $\wh C_{G,w}$ is still the $np$ multiplications
required to compute $\wh \beta_{g}$ for $g\in G$. 
We can also compute $\e( \wt C_{G,w})$ in about this amount
of work.  

The cost of computing $\var(\wt C_{G,w})$ by
a straightforward algorithm is
at least $np^2$, because we need  
$\bar X_{gh}$  and $\bar X_{gghh}$ for all $g,h\in G$.
Some parts of that computation can be sped up
to $O(np)$ by rewriting the expression as
described in Appendix 5. 
One of the terms however does not reduce to $O(np)$.
A straightforward implementation
costs $O(np^2)$ while an alternative expression costs $O(n^2p)$.
The latter is valuable in settings where the gene sets are
large compared to the sample size.
In the former case, the moment approximation has
cost comparable to $O(p^2)$ permutations.  If $n<p$
then the latter case is like $np$ permutations,
so the quadratic cost is comparable to on
the order of $p*\min(n,p)$ permutations.

We have verified our moment formulas by finding that
they match values found by enumerating all $n!$ permutations,
for some simulated data sets with small $n$.
During testing, we also compared permutation and our approximate
$p$-values on simulated data. We saw a close match but 
think that an illustration on real data is more compelling.
Section~\ref{sec:numerical} makes comparisons
using three genome wide expression studies in Parkinson's Disease (PD) patients:
\cite{moran2006whole}, \cite{scherzer} and \cite{zhang2}.


\section{Parkinson's Disease}\label{sec:numerical}
We illustrate our method using publicly available data from three expression studies
in Parkinson's Disease (PD) patients (Moran {\em et~al.}, 2006, 
Zhang {\em et~al.}, 2005, and Scherzer {\em et~al.}, 2007; Table~\ref{tab:expts}).
All three experiments contain genome wide expression values measured via 
a microarray experiment.  PD is a common neurodegenerative
disease; clinical symptoms often include rigidity, resting tremor and gait instability
\citep{abousleiman:2006}.  Pathologically, PD is characterized by neuronal-loss 
in the substantia nigra and the presence of $\alpha$-synuclein protein 
aggregates in neurons \citep{abousleiman:2006}.

\begin{table}[!t]
\centering
\begin{tabular}{lccc}\toprule
Reference & Tissue & \# Affected & \# Controls \\\midrule
Moran & Substantia nigra & 29 & 14\\
Zhang & Substantia nigra & 18 & 11 \\
Scherzer & Blood & 47 & 21 \\
\bottomrule
\end{tabular}
\caption{
Three data sets used for non-permutation GSEA\label{tab:expts}}
\end{table}

Using a selected set from the Broad Institute's mSigDB v3.1 \citep{subramanian:2005} 
and the presence
of PD as a response variable from the \cite{zhang2} dataset,  we  
visualized both permutation distributions and our approximation of 
these distributions (Figure~\ref{fig:mnexamp}).  
As discussed above, we use a linear test statistic, 
$\wh T_{G,w}=\sum_{g\in G}\hat\beta_g$, 
and a quadratic test statistic, 
$\wh C_{G,w} = \sum_{g\in G}\hat\beta_g^2$,
where $\hat\beta_g$ is a sample covariance between gene
expression and, in this case, disease status.
Figure~\ref{fig:mnexamp} shows these two test statistics
with a histogram of $99{,}999$ recomputations of those
statistics for permutations of treatment status versus gene
expression.   In principle, 
histograms of permuted test statistics
can be very complicated,
but in practice, they often resemble familiar parametric
distributions, as in Figure~\ref{fig:mnexamp}.

\begin{figure}[!tpb]
\centerline{\includegraphics[width=\hsize]{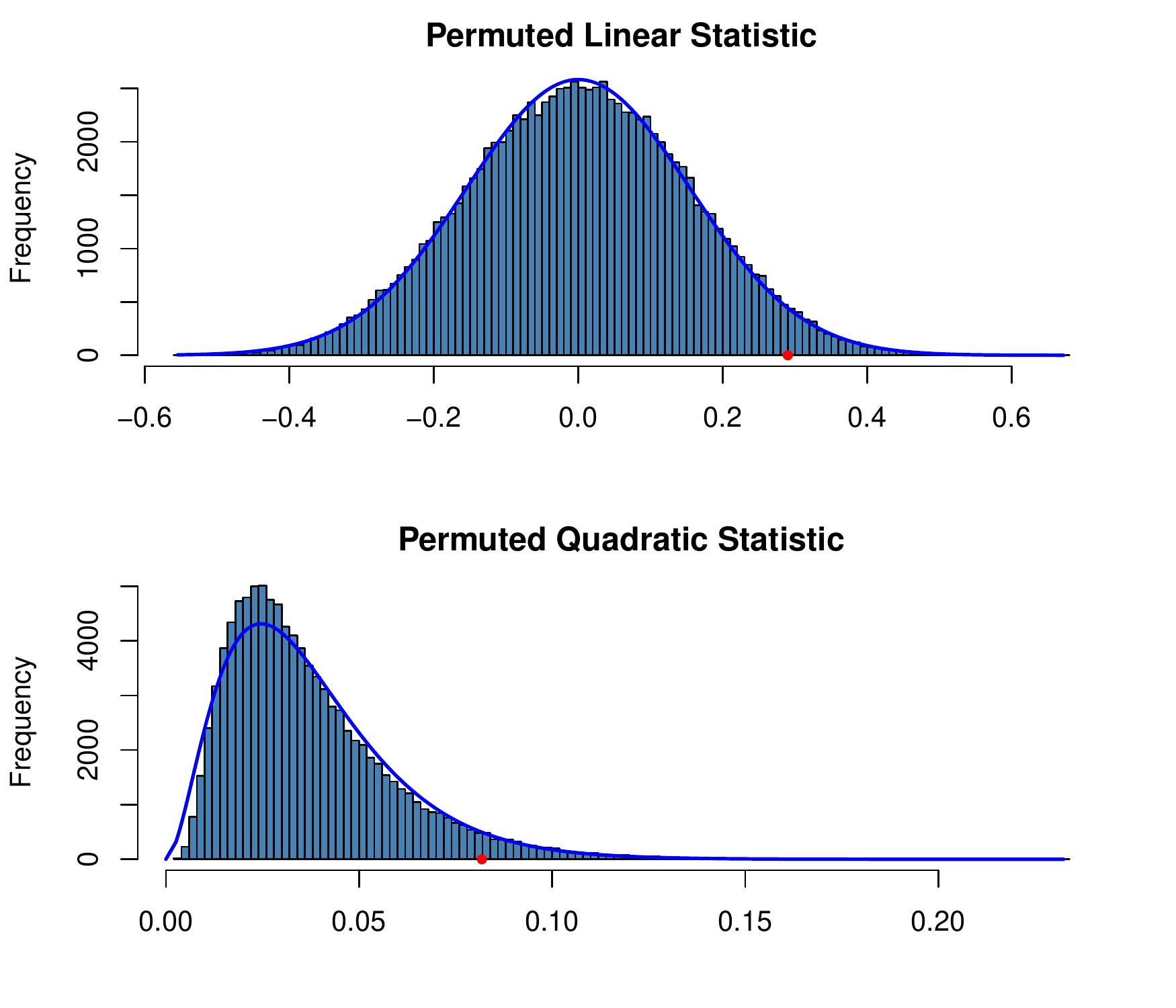}}
\caption{
Top panel shows a permutation histogram for a 
linear test statistic for the the steroid hormone signaling
pathway gene set as described in the text.
The bottom panel shows a quadratic test statistic.
Solid red dots indicate the observed values and curves
indicate parametric fits, based on normal and $\chi^2$
distributions.
}\label{fig:mnexamp}
\end{figure}

Using the fitted normal distribution to determine the rarity of the observed 
gene set statistic results in a two-tailed $p$-value of
$0.0604$ for the linear statistic while permutations
yield $p = 0.0595$. A fitted $\sigma^2\chi^2_{(\nu)}$ distribution
results in $p= 0.0425$ for the sum of squares gene set statistic, while permutations
yield $p = 0.0458$. The $p$-values are a quite close despite the
somewhat higher peak for the permutation histogram relative
to the $\chi^2$ density.

We compared our non-permutation $p$-values to $p$-values
for linear and quadratic statistics for the $6{,}303$ gene sets from 
mSigDB's curated gene sets and Gene Ontology 
(GO, Ashburner {\em et~al.}, 2000) gene sets collections (v3.1).  
One gene set was removed because it contained only one gene
in our experiments.  The average size of these gene sets is $79.40$ genes.  
For our gold standard we ran $999{,}999$ permutations
of the linear statistic and $499{,}999$ permutations
of the quadratic statistic.  For all of our permutations, we first calculated the 
observed test statistic for each of the $6{,}303$ gene sets and then permuted the 
$Y_i$'s $M$ times to obtain $6{,}303$ $\times$ $M$ permuted test statistics.  We next 
compared the pre-computed test statistic vector to our matrix of 
permuted test statistics.

For each set, we computed left-sided $p$-values, $p_L$, for the linear
statistic and two-sided $p$-values, $p_Q$, for the quadratic
statistic using these permutations.  We also computed the normal
and beta approximations of $p_L$ with our method.  (Figure~\ref{fig:corr}, left panel).    
We converted these one-sided $p$-values to two-sided $p$-values via 
$p=2\min(p_L, 1-p_L)$.  The beta approximation $p$-values 
are almost identical to the permutation $p$-values.  

\begin{figure}[!tpb]
\centerline{\includegraphics[width=0.9\hsize]{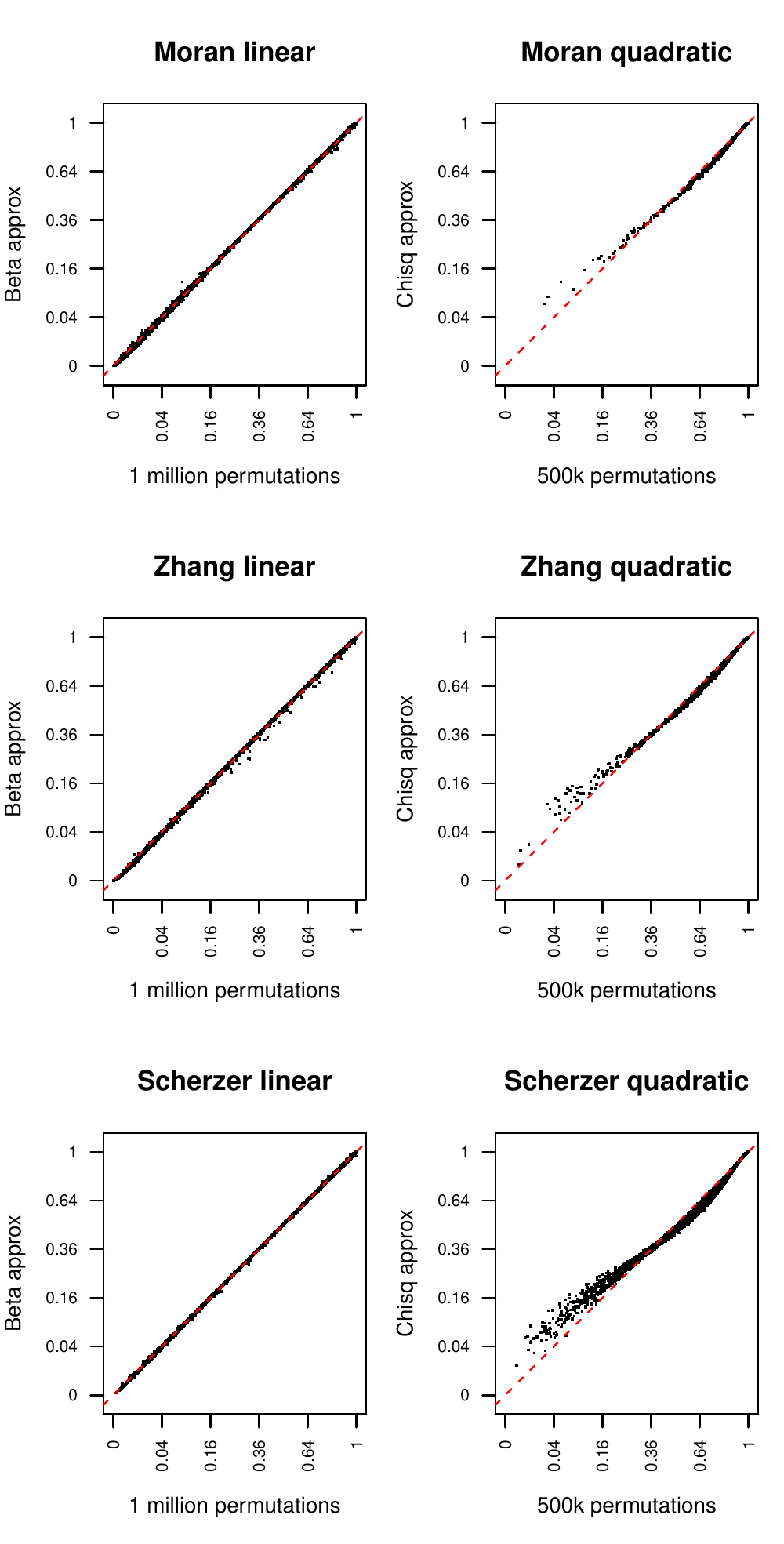}}
\caption{
Permutation $p$-values (x-axis) versus moment-based $p$-values (y-axis) 
for $6{,}303$ gene sets. The left column represents results for a 
linear test statistic, the right column for sum of squares. 
Data come from three genome-wide expression studies.  We applied the 
non-linear transformation $p^{1/2}$ to 
stretch the lower range of these distributions for a more informative visual.  
Red dotted line represents the line $y=x$.  
}\label{fig:corr}
\end{figure}

For our quadratic test statistic, we fit our
moment based $\sigma^2\chi^2_{(\nu)}$ approximation and computed
two-sided tailed $p$-values across all sets (Figure~\ref{fig:corr}, right panel).  
We see that the smallest $\chi^2$ non-permutation $p$-values
are slightly conservative. This may reflect the boundedness of
the permutation distribution combined with the unbounded 
right tail of the $\chi^2$ distribution.

In each of the three experiments, there is a tight correlation between 
the permutation-based $p$-values of all sets and both of our moment-based
methods (Table~\ref{tab:corrs}).  The beta and normal approximations are almost identical.
Our beta approximations are slightly closer to the gold standard 
than the normal approximations, but not by
a practically important amount. The beta approximation has
shorter tails than the Gaussian approximation. 
It yielded $p$-values somewhat
smaller than permutations did, while the Gaussian approximation
yielded $p$-values somewhat larger than the permutations did.
The $\chi^2$ approximations
also reproduce the ranking of the gold standard quite well,
though not as well as the normal and beta approximations to the linear statistic.  

\begin{table}[!t]
\centering
\begin{tabular}{lccccc}\toprule
Reference&  Normal $p_L$&  Beta $p_L$&  Normal $p_C$&  Beta $p_C$&  Chisq $p_Q$\\ \midrule
Moran   & 0.99991& 0.99997& 0.99973& 0.99991& 0.978\\
Zhang   & 0.99996& 0.99997& 0.99983& 0.99991& 0.990\\
Scherzer& 0.99998& 0.99999& 0.99991& 0.99997& 0.994\\ \bottomrule
\end{tabular}
\caption{Spearman correlations between gold standard 
($999{,}999$ and $499{,}999$ permutations for linear and quadratic statistics) 
and approximation $p$-values.  $p_L$ and $p_C$ represent results 
for one and two-tailed linear test statistics,
respectively. Chisq $p_Q$ represents results for the sum of squares analysis.\label{tab:corrs}}
\end{table}

For these data sets and $6{,}303$ gene sets, both of 
the linear statistics, 
which have more or less 
the same rank-ordering of $p$-values as $999{,}999$ permutations, 
could be approximated in about than the amount of time
it takes to compute $100$ permutations (Table~\ref{tab:timelin}, top block).
Our gene sets had an average size of about 80 genes. This lead
us to expect that the cost of the linear approximation
would be comparable to doing 80 permutations.  We
found that the Gaussian approximation cost about as
much as 100 permutations.  While this is a close match,
we remark that the time to do $M$ permutations is nearly
an affine function $a + bM$ with positive intercept $a$. At
such small $M$ the overhead costs dominated the total
cost making the per permutation costs hard to resolve.
The beta approximation was slightly slower than the
Gaussian one because it involves the sorting of the data.

\begin{table}
\centering
\begin{tabular}{llll}
\toprule
     Method& Moran& Zhang& Scherzer\\
\midrule
  $M=100$ &31.03& 29.84&    34.71\\
  $M=500$ &31.95& 32.49&    35.54\\
  $M=1{,}000{,}000$ &5010.17 &4434.77 &   3933.15\\
  Normal &29.74& 27.00 &   34.66\\
  Beta &30.79 &31.88 &   37.89\\
\midrule
   $M=30{,}000$ & 9146.27 & 7217.59   & 11808.02\\
   $M=40{,}000$ &12256.54 & 9636.06   & 16545.60\\
   $M=50{,}000$& 16833.08 &12564.06   & 21480.80\\
    $M=500{,}000$ &149588.37 &129667.73 &   187067.91\\
   $\chi^2$ &11020.62 &10600.82    &12677.15\\
\bottomrule
\end{tabular}
\caption{
Time in seconds for $p$-value calculations for $6{,}303$ gene sets in three 
genome-wide expression studies.  Linear statistic results with 
$M=100$,  $M=500$, and $M=1{,}000{,}000$ permutations, 
and the normal and beta approximations
are in the top block.  Timings for the quadratic statistic with 
$M=30{,}000$, $M=40{,}000$, $M=50{,}000$, and $M=500{,}000$ permutations, 
and the $\chi^2$ approximation
are presented in the bottom block.\label{tab:timelin} }
\end{table}

The $\chi^2$ approximation to the quadratic statistic 
has a computational cost about as much as $35{,}000$ to $45{,}000$ 
permutations, yet has a similar rank-ordering of $p$-values
$499{,}999$ permutations (Table~\ref{tab:timelin}, bottom block).
For the quadratic statistic we expected our algorithm
to cost as much as doing a number of permutations
equal to a small multiple of the mean square gene set size.
It cost about as much as $35{,}000$ to $45{,}000$ permutations
while the mean square set size was $27{,}171$.

After applying our permutation approximation methods to each
dataset in $6{,}303$ mSigDB
gene sets, we found many significantly enriched gene sets, even after correcting for
multiple testing (two-sided adjusted $p$-value $<$ 0.05).    
The most significantly enriched
sets are associated with metabolism and mitochondrial function, 
neuronal transmitters and serotonin, epigenetic modifications, and the transcription factor FOXP3
Supplemental Table 1\footnote{
\url{http://statweb.stanford.edu/~owen/reports/SupplementalTable1.xls}}
Each of these categories has some previously discovered association with PD,
although not through traditional gene set methods 
(metabolism and mitochondrial function: \cite{abousleiman:2006}; 
neuronal transmitters and serotonin: \cite{fox:2009}; 
epigenetic modifications: \cite{berthier:2013}; FOXP3: \cite{stone:2009}).  Through our 
new gene set enrichment method, we discovered a relationship between 
the expression of these gene sets and PD.

\section{Discussion}\label{sec:discussion}
Gene set methods are able to pool weak single gene signals over a set of genes to get
a stronger inference.  These methods and their corresponding permutation-based 
inferences are a staple of high throughput methods in genomics.  Because an experiment 
for this purpose may have a few to hundreds of microarrays or RNA-seq samples,
permutation can be computationally costly, and yet still result in granular $p$-values.  
In this paper, we introduce an approximation gene set method, which performs as well 
as permutation methods, in a fraction of the computation time and 
which generates   
continuous $p$-values.

Permutation methods have some valuable properties
that our approach does not share.
Permutation based inferences give exact $p$-values.
Our approximations are not ordinarily exact because
the permutation histogram is not in the parametric family
we use. 

The second advantage of permutations is that they apply to 
arbitrarily complicated statistics.  In our view, many of those complicated
statistics are much harder to interpret and are less intuitive 
than the plain sum and sum of
squared statistics we present.  Others have observed that simple linear and 
squared statistics outperform more complex approaches \citep{acke:stri:2009}.  
Our method allows for the weighting 
of coefficients in our statistics, granting users access to additional
useful and interpretable patterns.


Because of the disadvantages discussed above, there
has long been interest in finding approximations to permutation tests.
\cite{eden:yate:1933} noticed that the permutation distribution
closely matched a parametric distribution that one would get
running an $F$-test on the same data.  It has also been known since the 1940s that the permutation
distribution of the linear test is asymptotically normal
as $n$ increases~\citep{good:2004}.  More recently, \cite{knij:wess:rein:shmu:2009} approach 
the granularity
issue by taking a random sample of permutations and fitting a generalized
extreme value (GEV) distribution to the tail of their distribution.

Our work differs from these previous permutation approximation approaches. 
We use Gaussian or beta distributions
for the linear statistic and a $\chi^2$ distribution for the quadratic statistic. These
choices never place the observed test statistic strictly outside the possible
range of our reference distribution. In this way, we also avoid nonsensical $p$-values.

We have developed 
a new and intuitive method for gene set enrichment analysis that is computationally 
inexpensive, as accurate as permutation methods, and avoids the sample
granularity issue.  A Gaussian, beta, or $\chi^2$ approximation gives a principled way
to break ties among genes or gene sets whose test statistics
are larger than any seen in the $M$ permutations.   We applied our moment based 
approximations to three human Parkinson's Disease data sets and discovered the 
enrichment of several gene sets in this disease, none of which were  
mentioned in the original publications.

%
%


\section*{Acknowledgement}
We thank Nicholas Lewin-Koh, Joshua Kaminker,
Richard Bourgon, Sarah Kummerfeld, Thomas Sandmann,
and John Robinson for helpful comments.
ABO thanks Robert Gentleman, Jennifer Kesler and other
members of the Bioinformatics and Computational Biology Department 
at Genentech for their hospitality during his sabbatical there.

\paragraph{Funding:}
JLL is funded by Genentech, Inc.
ABO was supported by Genentech, Inc.\ and by Stanford
University while on a sabbatical.

\bibliographystyle{apalike}
\bibliography{gsea}

\begin{thebibliography}{}

\bibitem[Abou-Sleiman et~al., 2006]{abousleiman:2006}
Abou-Sleiman, P., Muqit, M., and Wood, N. (2006).
\newblock Expanding insights of mitochondrial dysfunction in parkinsonÕs
  disease.
\newblock {\em Nat Rev Neurosci}, 7:207--219.

\bibitem[Ackermann and Strimmer, 2009]{acke:stri:2009}
Ackermann, M. and Strimmer, K. (2009).
\newblock A general modular framework for gene set enrichment analysis.
\newblock {\em {BMC} Bioinformatics}, 10:47--66.

\bibitem[Anderson et~al., 1987]{ande:olki:unde:1987}
Anderson, T., Olkin, I., and Underhill, L. (1987).
\newblock Generation of random orthogonal matrices.
\newblock {\em SIAM Journal on Scientific and Statistical Computing},
  8(4):625--629.

\bibitem[Berthier and Pulido, 2013]{berthier:2013}
Berthier, A. JimŽnez-S‡inz, J. and Pulido, R. (2013).
\newblock Pink1 regulates histone h3 trimethylation and gene expression by
  interaction with the polycomb protein eed/wait1.
\newblock {\em Proc Natl Acad Sci USA}, 110(36):14729--34.

\bibitem[Bhatia and Davis, 2000]{bhat:chan:2000}
Bhatia, R. and Davis, C. (2000).
\newblock A better bound on the variance.
\newblock {\em The American Mathematical Monthly}, 107(4):353--357.

\bibitem[Eden and Yates, 1933]{eden:yate:1933}
Eden, T. and Yates, F. (1933).
\newblock On the validity of {Fisher's} $z$-test when applied to an actual
  sample of non-normal values.
\newblock {\em The Journal of Agricultural Science}, 23:6--7.

\bibitem[Fox et~al., 2009]{fox:2009}
Fox, S., Chuang, M., and Brotchie, J. (2009).
\newblock Serotonin and parkinsonÕs disease: On movement, mood, and madness.
\newblock {\em Movement Disorders}, 24(9):1255--1266.

\bibitem[Gentleman et~al., 2004]{gentleman:2004}
Gentleman, R., Carey, V., Bates, D., Bolstad, B., Dettling, M., Dudoit, S.,
  Ellis, B., Gautier, L., Ge, Y., Gentry, J., Hornik, K., Hothorn, T., Huber,
  W., Iacus, S., Irizarry, R., Leisch, F., Li, C., Maechler, M., Rossini, A.,
  Sawitzki, G., Smith, C., Smyth, G., Tierney, L., Yang, J., and Zhang, J.
  (2004).
\newblock Bioconductor: open software development for computational biology and
  bioinformatics.
\newblock {\em Genome Biol}, 5(10):R80.1--R80.16.

\bibitem[Goeman and B\"uhlmann, 2007]{goem:buhl:2007}
Goeman, J.~J. and B\"uhlmann, P. (2007).
\newblock Analyzing gene expression data in terms of gene sets: methodological
  issues.
\newblock {\em Bioinformatics}, 23(8):980--987.

\bibitem[Good, 2004]{good:2004}
Good, P.~I. (2004).
\newblock {\em Permutation, parametric, and bootstrap tests of hypotheses}.
\newblock Springer, New York.

\bibitem[Jiang and Gentleman, 2007]{jian:gent:2007}
Jiang, Z. and Gentleman, R. (2007).
\newblock Extensions to gene set enrichment.
\newblock {\em Bioinformatics}, 23(3):306--313.

\bibitem[Knijnenburg et~al., 2009]{knij:wess:rein:shmu:2009}
Knijnenburg, T.~A., Wessels, L. F.~A., Reinders, M. J.~T., and Shmulevich, I.
  (2009).
\newblock Fewer permutations, more accurate p-values.
\newblock {\em Bioinformatics}, 25(12):i161--i168.

\bibitem[Langsrud, 2005]{lang:2005}
Langsrud, O. (2005).
\newblock Rotation tests.
\newblock {\em Statistics and computing}, 15:53--60.

\bibitem[Lehmann and Romano, 2005]{lehm:roma:2005}
Lehmann, E.~L. and Romano, J.~P. (2005).
\newblock {\em Testing statistical hypotheses}.
\newblock Springer.

\bibitem[Mootha et~al., 2003]{moothetal:2003}
Mootha, V.~K., Lindgren, C.~M., Eriksson, K.~F., Subramanian, A., Sihag, S.,
  Lehar, J., Puigserver, P., Carlsson, E., Ridderstrale, M., Laurila, E.,
  Houstis, N., Daly, M.~J., Patterson, N., Mesirov, J.~P., Golub, T.~R.,
  Tamayo, P., Spiegelman, B., Lander, E.~S., Hirschhorn, J.~N., Altshuler, D.,
  and Groop, L.~C. (2003).
\newblock {PGC}-1$\alpha$-responsive genes involved in oxidative
  phosphorylation are coordinately downregulated in human diabetes.
\newblock {\em Nature Genetics}, 34:267--273.

\bibitem[Moran et~al., 2006]{moran2006whole}
Moran, L.~B., Duke, D.~C., Deprez, M., Dexter, D.~T., Pearce, R. K.~B., and
  Graeber, M.~B. (2006).
\newblock Whole genome expression profiling of the medial and lateral
  substantia nigra in {Parkinson's} disease.
\newblock {\em Neurogenetics}, 7(1):1--11.

\bibitem[Newton et~al., 2007]{newt:quin:fern:denb:seng:ahlq:2007}
Newton, M.~A., Quintana, F.~A., {den Boon}, J.~A., Sengupta, S., and Ahlquist,
  P. (2007).
\newblock Random-set methods identify distinct aspects of the enrichment signal
  in gene-set analysis.
\newblock {\em The Annals of Applied Statistics}, pages 85--106.

\bibitem[Owen, 2005]{owen:2005}
Owen, A.~B. (2005).
\newblock Variance of the number of false discoveries.
\newblock {\em Journal of the Royal Statistical Society, Series B},
  67(3):411--426.

\bibitem[Scherzer et~al., 2007]{scherzer}
Scherzer, C.~R., AC, A. C.~E., Morse, L.~J., Liao, Z., Locascio, J.~J., Fefer,
  D., Schwarzschild, M.~A., Schlossmacher, M.~G., Hauser, M.~A., Vance, J.~M.,
  Sudarsky, L.~R., Standaert, D.~G., Growdon, J.~H., Jensen, R.~V., and
  Gullans, S.~R. (2007).
\newblock Molecular markers of early {Parkinson's} disease based on gene
  expression in blood.
\newblock {\em Proc Natl Acad Sci}, 104(3):955--60.

\bibitem[Smyth, 2005]{smyth:2005}
Smyth, G. (2005).
\newblock Limma: linear models for microarray data.
\newblock In Gentleman, R., Carey, V., Dudoit, S., Irizarry, R., and Huber, W.,
  editors, {\em Bioinformatics and Computational Biology Solutions Using {R}
  and Bioconductor}, pages 397--420. Springer, New York.

\bibitem[Stone et~al., 2009]{stone:2009}
Stone, D., Reynolds, A., Mosely, R., and Gendelman, H. (2009).
\newblock Innate and adaptive immunity for the pathobiology of parkinson's
  disease.
\newblock {\em Antioxid Redox Signal}, 11(9):2151--2166.

\bibitem[Subramanian et~al., 2005]{subramanian:2005}
Subramanian, A., Tamayo, P., Mootha, V., Mukherjee, S., Ebert, B., Gillette,
  M., Paulovich, A., Pomeroy, S., Golub, T., Lander, E., and Mesirov, J.
  (2005).
\newblock Gene set enrichment analysis: a knowledge-based approach for
  interpreting genome-wide expression profiles.
\newblock {\em Proc Natl Acad Sci USA}, 102(43):15545--50.

\bibitem[Wedderburn, 1975]{wedd:1975}
Wedderburn, R. W.~M. (1975).
\newblock Random rotations and multivariate normal simulation.
\newblock Technical report, Rothamsted Experimental Station.

\bibitem[Wu et~al., 2010]{wu:lim:vail:asse:visv:smyt:2010}
Wu, D., Lim, E., Vaillant, F., Asselin-Labat, M.-L., Visvader, J.~E., and
  Smyth, G.~K. (2010).
\newblock Roast: rotation gene set tests for complex microarray experiments.
\newblock {\em Bioinformatics}, 26(17):2176--2182.

\bibitem[Zhang et~al., 2005]{zhang2}
Zhang, Y., James, M., Middleton, F.~A., and Davis, R.~L. (2005).
\newblock Transcriptional analysis of multiple brain regions in {Parkinson's}
  disease supports the involvement of specific protein processing, energy
  metabolism, and signaling pathways, and suggests novel disease mechanisms.
\newblock {\em Am J Med Genet B Neuropsychiatr Genet}, 137B(1):5--16.

\bibitem[Zhou et~al., 2009]{zhou:wang:wang:2009}
Zhou, C., Wang, H.~J., and Wang, Y.~M. (2009).
\newblock Efficient moments-based permutation tests.
\newblock {\em Advances in neural information processing systems}, 22:2277.

\end{thebibliography}

\newcommand{\lemmomentone}{1} 
\newcommand{\lemmomenttwo}{2} 
\newcommand{\lemrotateone}{3} 
\newcommand{\numberofpreviouslemmas}{3} 

\renewcommand{\le}{\leqslant}
\renewcommand{\ge}{\geqslant}

\newcommand{\dustd}{\mathsf{U}}
\newcommand{\diag}{\mathrm{diag}}
\newcommand{\tr}{\mathrm{tr}}

\subsection*{Appendix 1: Proof of Lemma~\lemmomentone}
This appears in~\cite{owen:2005} but we prove it here to keep
the paper self-contained.
First
\begin{align*}
n^2\e(\wt\beta_g\wt\beta_h) = 
\sum_i\sum_{i'}X_{gi}X_{hi'}\e(\wt Y_i\wt Y_{i'})
\end{align*}
Recall that $\mu_2 =\frac1n\sum_{i=1}^n Y_i^2$.
Then
$$\e(\wtyi\wtyip) =
\begin{cases}
\mu_2,& i'=i\\
-\dfrac1{n-1}\mu_2, & i'\ne i
\end{cases}
$$
and so
\begin{align*}
n^2\e(\wt\beta_g\wt\beta_h) &= 
\sum_i\sum_{i'}X_{gi}X_{hi'}\e(\wt Y_i\wt Y_{i'})\\
& = 
\mu_2\sum_i\sum_{i'}X_{gi}X_{hi'}
\Bigl( 1_{i=i'} - \frac1{n-1}1_{i\ne i'}\Bigr)\\
& = 
\mu_2\sum_i\sum_{i'}X_{gi}X_{hi'}
\Bigl( \frac{n}{n-1}1_{i=i'} - \frac1{n-1}\Bigr)\\
& = \frac{n}{n-1}\mu_2\sum_iX_{gi}X_{hi}\\
& \equiv \frac{n^2}{n-1}\mu_2\bar X_{gh},
\end{align*}
proving Lemma~\lemmomentone.\ $\Box$

\subsection*{Appendix 2: Proof of Lemma~\lemmomenttwo}

The fourth moment contains terms of the form
$$X_{gi}X_{hj}X_{rk}X_{s\ell}\e( \wtyi\wtyj\wtyk\wtyl)$$
and there are different special cases depending on 
which pairs of indices among $i$, $j$, $k$ and $\ell$
are equal.
We need the following fourth moments of $Y$
in which all indices are distinct:
\begin{align*}
\mu_{4k} &= \e( \wt Y_i^4 )\\
\mu_{3k} & = \e( \wt Y_i^3\wt Y_j)\\
\mu_{2p} & = \e( \wt Y_i^2\wt Y_j^2)\\
\mu_{1p} & = \e( \wt Y_i^2\wt Y_j\wt Y_k)\\
\mu_\emptyset & = \e( \wt Y_i\wt Y_j\wt Y_k\wtyl),
\end{align*}
and where the subscripts are mnemonics for terms
four of a kind, three of a kind, two pair, one pair
and nothing special.

We can express all of these moments in terms
of $\mu_2$ and $\mu_4 = (1/n)\sum_{i=1}^nY_i^4$.
Each moment is a normalized sum over distinct
indices. We can write these in terms of normalized sums
over all indices. Many of those terms vanish because
$\sum_iY_i=0$.

Let $\sum^*$ represent summation over distinct
indices,  as in
\begin{align*}
\sums_{ij} f_{ij}
&= \sum_{i=1}^n\sum_{j=1,j\ne i}^n f_{ij},\\
\sums_{ijk} f_{ijk}& = \sum_{i=1}^n\sum_{j=1,j\ne i}^n
\sum_{k=1,k\ne i,k\ne j} f_{ijk}
\end{align*}
and so on. We can write these sums in terms of 
unrestricted sums:
\begin{align*}
\sums_{ij} f_{ij} & =  \sum_{ij}f_{ij} - \sum_i f_{ii}\\
\sums_{ijk} f_{ijk} & = \sum_{ijk}f_{ijk}
 - \sum_{ij}(f_{iij}+f_{iji}+f_{ijj}) + 2\sum_if_{iii},\quad\text{and}\\
\sums_{ijk\ell}f_{ijk\ell}
& = \sum_{ijk\ell}f_{ijk\ell}-\sum_{ijk}\Bigl( 
f_{ijki} + f_{ijkj} + f_{ijkk} + f_{ijik} + f_{ijjk} + f_{iijk}
\Bigr)\\
& \pe
+\sum_{ij}\Bigl( 
2( f_{ijjj} + f_{ijii} + f_{iiji} + f_{iiij})
+ f_{ijij} + f_{ijji} + f_{iijj} \Bigr)
- 6\sum_i f_{iiii}.
\end{align*} 
See Gleich and Owen (2011) for details.

We will use the last expression in a context
where $f_{ijk\ell}$ vanishes when summed over 
the entire range of any one of its indices.
In that case
\begin{align}\label{eq:fourhom}
\sums_{ijk\ell}f_{ijk\ell}
& = 
\sum_{ij}\Bigl(
f_{ijij} + f_{ijji} + f_{iijj} \Bigr)- 6\sum_i f_{iiii}.
\end{align}
We also use the notation $n^{(k)} = n(n-1)(n-2)\cdots(n-k+1)$,
often called `$n$ to $k$ factors', where $k$ is a positive integer.
Now
\begin{align*}
\mu_{4k} &= \frac1n\sum_{i=1}^nY_i^4 = \mu_4,\\
\mu_{3k} & = \frac1{n^{(2)}} \sums_{ij} Y_i^3Y_j
= \frac1{n^{(2)}} \biggl(\sum_{ij} Y_i^3Y_j
-\sum_iY_i^4\biggr)\\
& = -\frac{\mu_4}{n-1},\\
\mu_{2p} & = \frac1{n^{(2)}}\sums_{ij}Y_i^2Y_j^2
= \frac1{n^{(2)}}\biggl(\sum_{ij}Y_i^2Y_j^2
-\sum_i Y_i^4\biggr)\\
&
= \frac1{n-1}\bigl(n\mu_2^2-\mu_4\bigr),\quad\text{and}\\
\mu_{1p} &= \frac1{n^{(3)}}\sums_{ijk}Y_i^2Y_jY_k\\
&= \frac1{n^{(3)}}\biggl(
\sum_{ijk} Y_i^2Y_jY_k
-\sum_{ij}\bigl( 2Y_i^3Y_j + Y_i^2Y_j^2\bigr) +2\sum_iY_i^4
\biggr)\\
& = 
\frac{ -n\mu_2^2+2\mu_4}{(n-1)(n-2)}.
\end{align*}

Finally using~\eqref{eq:fourhom}, $n^{(4)}\mu_\emptyset$ equals
\begin{align*}
\sums_{ijk\ell}Y_iY_jY_kY_\ell
& = 3\sum_{ij} Y_i^2Y_j^2 - 6\sum_iY_i^4
= 3n^2\mu_2^2-6n\mu_4
\end{align*}
so that
$$
\mu_\emptyset =  \frac1{(n-1)(n-2)(n-3)}\bigl( 3n\mu_2^2-6\mu_4\bigr).
$$
We may summarize these results via
\begin{align*}
\begin{pmatrix}
\mu_{4k}\\
\mu_{3k}\\
\mu_{2p}\\
\mu_{1p}\\
\mu_\emptyset
\end{pmatrix}
=
A\begin{pmatrix}
\mu_2^2\\
\mu_4
\end{pmatrix},
\end{align*}
where the matrix $A$ is given in the statement of Lemma~\lemmomenttwo.

Now
\begin{align*}
n^4\e(\wt\beta_g\wt\beta_h\wt\beta_r\wt\beta_s)
& =
\sum_{ijk\ell} X_{gi}X_{hj}X_{rk}X_{s\ell}\e( \wtyi\wtyj\wtyk\wtyl )\\
& = \mu_{4k}\sum_i X_{gi}X_{hi}X_{ri}X_{si}\\
&\pe+\mu_{3k}\sums_{ij}\bigl(
 X_{gi}X_{hi}X_{ri}X_{sj}
+X_{gi}X_{hi}X_{rj}X_{si}
+X_{gi}X_{hj}X_{ri}X_{si}
+X_{gj}X_{hi}X_{ri}X_{si}
\bigr)\\
&\pe+\mu_{2p}
\sums_{ij}\bigl( 
 X_{gi}X_{hi}X_{rj}X_{sj}
+X_{gi}X_{hj}X_{ri}X_{sj}
+X_{gi}X_{hj}X_{rj}X_{si}\bigr)\\
&\pe +\mu_{1p}
\sums_{ijk}\bigl(
  X_{gi}X_{hi}X_{rj}X_{sk}
+ X_{gi}X_{hj}X_{ri}X_{sk}
+ X_{gi}X_{hj}X_{rk}X_{si}\\
&\pe\phantom{\mu_{1p}\sums_{ijk}}
+ X_{gi}X_{hj}X_{rj}X_{sk}
+ X_{gi}X_{hj}X_{rk}X_{sj}
+ X_{gi}X_{hj}X_{rk}X_{sk}
\bigr)\\
&\pe+\mu_\emptyset\sums X_{gi}X_{hj}X_{rk}X_{s\ell}.
\end{align*}
Next, we write the terms of $n^4\e(\wt\beta_g\wt\beta_h\wt\beta_r\wt\beta_s)$
using $\bar X_{ghrs}$ and similar moments.

The coefficient of $\mu_{4k}$ is $\sum_iX_{gi}X_{hi}X_{ri}X_{si} = 
n\bar X_{ghrs}$.
The coefficient of $\mu_{3k}$ contains
$$
\sums_{ij} X_{gi}X_{hi}X_{ri}X_{sj} = 
\sum_{ij} X_{gi}X_{hi}X_{ri}X_{sj} 
- \sum_{i} X_{gi}X_{hi}X_{ri}X_{si}  
= -n\bar X_{ghrs}
$$
and after summing all four such terms, the coefficient is
$-4n\bar X_{ghrs}$.
The coefficient of $\mu_{2p}$ contains
\begin{align*}
\sums_{ij}
 X_{gi}X_{hi}X_{rj}X_{sj}
= \sum_{ij} X_{gi}X_{hi}X_{rj}X_{sj}
-\sum_i  X_{gi}X_{hi}X_{ri}X_{si}
= -n\bar X_{ghrs}
\end{align*}
and accounting for all three terms yields $-3n\bar X_{ghrs}$.

The coefficient of $\mu_{1p}$ contains
\begin{align*}
\sums_{ijk}X_{gi}X_{hi}X_{rj}X_{sk} 
& =
\sum_{ijk}X_{gi}X_{hi}X_{rj}X_{sk}  
-\sum_{ij}X_{gi}X_{hi}X_{ri}X_{sj}  \\
&\pe
-\sum_{ik}X_{gi}X_{hi}X_{rj}X_{si} 
-\sum_{jk}X_{gi}X_{hi}X_{rj}X_{sj}  
+2\sum_i X_{gi}X_{hi}X_{ri}X_{si} \\
& = -n^2\bar X_{gh}\bar X_{rs} +2n\bar X_{ghrs}.
\end{align*}
Summing all $6$ terms, we find that the coefficient is
$$
-2n^2(
 \bar X_{gh}\bar X_{rs}
+\bar X_{gr}\bar X_{hs}
+\bar X_{gs}\bar X_{hr})
+12n\bar X_{ghrs}.
$$

The coefficient of $\mu_\emptyset$ is, using~\eqref{eq:fourhom},
\begin{align*}
\sums_{ijk\ell}X_{gi}X_{hj}X_{rk}X_{s\ell}
& = 
\sum_{ij}\Bigl(
X_{gi}X_{hj}X_{ri}X_{sj}
+X_{gi}X_{hj}X_{rj}X_{si}
+X_{gi}X_{hi}X_{rj}X_{sj}
\Bigr) \\
&\pe- 6\sum_iX_{gi}X_{hi}X_{ri}X_{si}\\
& = n^2\bigl( \bar X_{gh}\bar X_{rs}+\bar X_{gr}\bar X_{hs}+\bar X_{gs}\bar X_{hr}\bigr)-6n\bar X_{ghrs}.
\end{align*}

We may summarize these results via
\begin{align*}
\e(\wt\beta_g\wt\beta_h\wt\beta_r\wt\beta_s)
=
\begin{pmatrix}
\mu_{4k}\\
\mu_{3k}\\
\mu_{2p}\\
\mu_{1p}\\
\mu_\emptyset
\end{pmatrix}^\tran
B
\begin{pmatrix}
\bar X^*_{ghrs}/n^2\\[1ex]
\bar X_{ghrs}/n^3
\end{pmatrix},
\quad\text{for}\quad
B = 
\begin{pmatrix}
\pn0 & \pn1\\
\pn0 & -4\\
\pn1 & -3\\
-2   & \,12\\
\pn1 & -6
\end{pmatrix},
\end{align*}
where $\bar X^*_{gh,rs}  = 
\bar X_{gh}\bar X_{rs}+\bar X_{gr}\bar X_{hs}+\bar X_{gs}\bar X_{hr}$,
completing the proof of Lemma~\lemmomenttwo.

\subsection*{Appendix 3: moments of orthogonal random matrix elements.}

We will need low order moments of orthogonal random matrices
to study the moments of linear and quadratic test statistics
under rotation sampling.

For integers $n\ge k\ge 1$, 
let $V_{n,k} = \{ Q\in\real^{n\times k}\mid Q^\tran Q = I_k\}$,
known as the Stiefel manifold.
We will make use of the uniform distributions on $V_{n,k}$.
There is a natural identification of $V_{n,1}$ with the
unit sphere.

Let $Q\in\real^{n\times n}$ be a uniform random rotation matrix.
This implies, among other things, that each column of $Q$ is
a uniform random point on the unit sphere in $n$ dimensions.

By symmetry, we find that $\e( Q_{ij} ) = 0$.
Similarly $\e( Q_{ij}^2 ) = \e( (1/n)\sum_{j=1}^n Q_{ij}^2)=1/n$
and $\e( Q_{ij}Q_{rs} ) =0$ unless $i=r$ and $j=s$.

\cite{ande:olki:unde:1987}
give 
\begin{align}\label{eq:4th}
\e( Q_{ij}^4) = \frac3{n(n+2)}. 
\end{align}
We are interested in all fourth moments
$\e( Q_{ij}Q_{k\ell}Q_{rs}Q_{tu})$ of $Q$.
If any of $j,\ell,s,u$ appears exactly once
then the fourth moment is $0$ by symmetry. 
To see this, suppose that index $\ell$
appears exactly once.
Now define the matrix $\wt Q$ with elements
$$
\wt Q_{ij} = \begin{cases}
-Q_{ij} & j=\ell,\\
Q_{ij} & j\ne \ell.
\end{cases}
$$
If $Q\sim\dustd(V_{n,n})$ then $\wt Q\sim\dustd(V_{n,n})$ too
by invariance of $\dustd(V_{n,n})$ to multiplication on the
right by the orthogonal matrix $\diag(1,1,\dots,1,-1,1,\dots,1)$,
with a $-1$ in the $j'$th position.
Then
\begin{align*}
\e( Q_{ij}Q_{k\ell}Q_{rs}Q_{tu} )
&= \frac12\e\bigl(
Q_{ij}Q_{k\ell}Q_{rs}Q_{tu} + \wt Q_{ij}\wt Q_{k\ell}\wt Q_{rs}\wt Q_{tu}\bigr)\\
&= \frac12\e\bigl(
Q_{ij}Q_{k\ell}Q_{rs}Q_{tu} + Q_{ij}(-Q_{k\ell})Q_{rs}Q_{tu}\bigr)\\
&=0.
\end{align*}
Similarly, because $Q^\tran$ is also uniformly distributed
on $V_{n,n}$ we find that if any of $i,k,r,t$ appear
exactly once the moment is zero. If one index appears exactly
three times, then some other moment must appear exactly once.
As a result, the only nonzero fourth moments are products of
squares and pure fourth moments.
Their values are given in the Lemma below.

\begin{lemma}
Let $Q\sim\dustd(V_{n,n})$. Then
$$
\e( Q_{ij}^2Q_{rs}^2 ) =
\begin{cases}
\dfrac{3}{n(n+2)}, & i=r\ \&\ j=s\\[2ex]
\dfrac{1}{n(n+2)}, & 1_{i=r} + 1_{j=s} = 1\\[2ex]
\dfrac{n+1}{n(n-1)(n+2)}, & i\ne r\ \&\ j\ne s.
\end{cases}
$$
\end{lemma}
\begin{proof}
The first case was given by \cite{ande:olki:unde:1987}.

For the second case, there is no loss of generality
in computing $\e( Q_{11}^2Q_{21}^2 )$.
The vector $(Q_{11},Q_{21},\dots,Q_{n1})$ is uniformly distributed
on the sphere. Given $Q_{11}$, the point
$(Q_{21},Q_{31},\dots,Q_{n1})$ is uniformly distributed on
the $n-1$ dimensional sphere of radius $\sqrt{1-Q_{11}^2}$.
Therefore $\e( Q_{21}^2 \mid Q_{11} ) = (1-Q_{11}^2)/(n-1)$
and so
\begin{align*}
\e(Q_{11}^2Q_{21}^2) &= \frac1{n-1}\e( Q_{11}^2-Q_{11}^4)
=\frac1{n-1}\left( \frac1n - \frac3{n(n+2)}\right)
=\frac1{n(n+2)}.
\end{align*}

For the remaining case we let
$\theta = \e( Q_{ij}^2Q_{rs}^2)$
for $i\ne r$ and $j\ne s$.
Summing over $n^4$ combinations of indices we
find that 
$$\sum_{i=1}^n\sum_{j=1}^n\sum_{r=1}^n\sum_{s=1}^nQ_{ij}^2Q_{rs}^2
=\biggl( \sum_{ij}Q_{ij}^2\biggr)^2 = n^2
$$
by orthogonality of $Q$.
Therefore
$$
n^2 = \e\biggl( \sum_{ij}\sum_{rs}Q_{ij}^2Q_{rs}^2 \biggr)
= n^2 \e( Q_{11}^4 )
+ 2n^2(n-1)\e( Q_{11}^2Q_{12}^2 )
+n^2(n-1)^2\theta.
$$
Solving for $\theta$ we get
\begin{align*}
\theta = \dfrac
{n^2- \frac{3n}{n+2} - \frac{2n(n-1)}{n+2}}
{n^2(n-1)^2}
= \frac{n+1}{n(n-1)(n+2)}.\qquad\qedhere
\end{align*}
\end{proof}

\subsection*{Appendix 4: proof of Lemma~\lemrotateone.}

Let $X_i\in \real^p$ where $p=|G|$ and $Y_i\in\real$ for 
$i=1,\dots,n$.
Both $X_i$ and $Y_i$ are centered:
$\sum_iX_i=0$ and $\sum_i Y_i=0$.

The sample coefficients for genes $g\in G$ are given by the 
vector $\hat\beta = L=(1/n)\sum_iX_iY_i$.
The reference distribution is formed by sampling
values of $\wt\beta = (1/n)\sum_iX_i\wt Y_i$ where $\wt Y$ is a rotated
version of $Y$.

The rotation is one that preserves the mean of $Y$
while rotating in the $n-1$ dimensional space of contrasts.
As in \cite{lang:2005}, we
let $W\in\real^{n\times(n-1)}$ be any fixed contrast matrix
satisfying $W^\tran W=I_{n-1}$ and $W^\tran 1_n=0_{n-1}$.
Then the rotated version of $Y$ is
$$
\wt Y =  W Q W^\tran Y,\quad\text{where}\quad Q\sim\dustd(V_{n-1,n-1})$$
is a uniform random $n-1$ dimensional rotation matrix.

It is convenient to introduce centered
quantities
$X^c = W^\tran X\in\real^{(n-1)\times p}$,
$Y^c = W^\tran Y\in\real^{n-1}$
and 
$\wt Y^c = W^\tran \wt Y\in\real^{n-1}$.
These sum to zero even when $X$, $Y$  and $\wt Y$ do not.
Their main difference from those variables is that they
have $n-1$ rows, not $n$.

Now $\wt\beta = (1/n)X^\tran \wt Y
=(1/n) X^\tran WQW^\tran Y=(1/n){X^c}^\tran Q Y^c$, so
\begin{align*}
\e( \wt\beta ) = (1/n){X^c}^\tran\e(Q){Y^c}^\tran = 0.
\end{align*}
For the rest of the proof, we
need the covariance matrix of $\wt \beta$.
Now
\begin{align}
\e( \wt\beta\wt\beta^\tran )
&=\frac1{n^2}
{X^c}^\tran \e\bigl( Q^\tran Y^c{Y^c}^\tran Q\bigr){X^c}^\tran
 = \frac1{n^2} {X^c}^\tran\e\bigl( Q^\tran Z Q\bigr)X^c\notag
\end{align}
where $Z = Y^c{Y^c}^\tran \in \real^{(n-1)\times(n-1)}$.

The $ij$ element of $Q^\tran ZQ$ is
$
(Q^\tran ZQ)_{ij} 
 = \sum_{k=1}^{n-1} \sum_{\ell=1}^{n-1} Z_{k\ell} Q_{ki} Q_{\ell j}
$
which has expected value
$$
\sum_{k=1}^{n-1} \sum_{\ell=1}^{n-1} Z_{k\ell} 1_{k=\ell}1_{i=j}/(n-1)
=\frac{1_{i=j}}{n-1}\sum_{k=1}^{n-1}Z_{kk} = 1_{i=j}\frac{n}{n-1}\mu_2
$$
where $\mu_2 = (1/n)\sum_{i=1}^nY_i^2 = (1/n)\sum_{i=1}^n{Y^c_i}^2$.
That is
$$
\e( Q^\tran ZQ) = \frac{n\mu_2}{n-1}I_{n-1}
$$
and so
$$
\e(\wt\beta\wt\beta^\tran) = \frac{\mu_2}{n(n-1)}{X^c}^\tran X^c. 
$$
In particular $\e(\wt\beta_g\wt\beta_h) = \e(\wt\beta\wt\beta^\tran)_{gh}
= \bar X_{gh}\mu_2/(n-1)$, matching the value under permutation.

\subsection*{Appendix 5: cost analysis of $\var(\wt C_{G,w})$}

Recall from Corollary 2 that
in an experiment with $n\ge 4$ and genes $g,h$,
\begin{align*}
\cov(\wt\beta^2_g,\wt\beta^2_h) 
&=
\begin{pmatrix} \mu_2^2\\\mu_4\end{pmatrix}^\tran 
A^\tran B
\begin{pmatrix}
\bar X^*_{gghh}
/n^2\\[1ex]
\bar X_{gghh}/n^3
\end{pmatrix} -\frac{\mu_2^2}{(n-1)^2}\bar X_{gg}\bar X_{hh},
\end{align*}
where $\bar X^*_{gghh} = \bar X_{gg}\bar X_{hh} + 2\bar X_{gh}^2$
and $A^\tran B$ is a given $2\times 2$ matrix.

To compute
\begin{align*}
\var( \wt C_{G,w}) = \sum_{g\in G}\sum_{h\in G}
w_gw_h\cov(\wt \beta_g^2,\wt\beta_h^2)
\end{align*}
we need $\mu_2$, $\mu_4$ and $A^\tran B$ which are
very inexpensive. We also need
\begin{align*}
S_1\equiv \sum_{g\in G}\sum_{h\in G}w_gw_h\bar X_{gg}\bar X_{hh}
= \Biggl(\, \sum_{g\in G}w_g\bar X_{gg}\Biggr)^2.
\end{align*}
By expressing $S_1$ as a square, we
find that it can be computed in $O(np)$ work,
not $O(np^2)$ which a naive implementation would provide.
We can compute all of the $\bar X_{gg}$'s
in $np$ multiplications and this is the largest part of the cost.
If gene $g$ belongs to many gene sets $G$ we only
need to compute $\bar X_{gg}$ once and so the cost per additional
gene set could be lower.

A similar analysis yields that
$$
S_2\equiv\sum_{g\in G}\sum_{h\in G}w_gw_h\bar X_{gghh}
= \frac1n\sum_{i=1}^n \Biggl(\, \sum_{g\in G}w_gX_{gi}^2\Biggr)^2
$$
is also an $O(np)$ computation. Unfortunately
$S_3 \equiv \sum_{g\in G}\sum_{h\in G}\bar X_{gh}^2$
does not reduce to an $O(np)$ computation.
As written it costs $O(np^2)$.
In cases where $p>n$, we can however reduce
the cost to $O(n^2p)$ via
\begin{align*}
S_3
& = \sum_{g\in G}\sum_{h\in G}w_gw_h\Biggl(\,\frac1{n}\sum_{i=1}^nX_{gi}X_{hi}\Biggr)^2
 = \frac1{n^2}\sum_{g\in G}\sum_{h\in G}w_gw_h
\sum_{i=1}^n
X_{gi}X_{hj}\sum_{j=1}^nX_{gj}X_{hj}\\
& = 
\frac1{n^2}
\sum_{i=1}^n\sum_{j=1}^n
\Biggl(\,\sum_{g\in G}w_gX_{gi}\Biggr)^2.
\end{align*}

In terms of these sum quantities,
\begin{align*}
\var( \wt C_{G,w}) 
&=
\begin{pmatrix} \mu_2^2\\\mu_4\end{pmatrix}^\tran 
A^\tran B
\begin{pmatrix}
(S_1 + 2S_3)/n^3\\[1ex]
S_2/n^3
\end{pmatrix} -\frac{\mu_2^2}{(n-1)^2}S_1.
\end{align*}

\end{document}